\documentclass{llncs}
\usepackage{amsmath,amssymb}
\usepackage{graphicx}
\usepackage{tabularx}
\usepackage{subfig}

\spnewtheorem{fact}[theorem]{Fact}{\bfseries}{\itshape}
\spnewtheorem{observation}[theorem]{Observation}{\bfseries}{\itshape}
\spnewtheorem*{myremark}{Remark}{\bfseries}{\itshape}

\newcommand{\Oh}{\mathcal{O}}

\begin{document}
\title{On Maximal Unbordered Factors}
\author{
Alexander Loptev\inst{1}
\and 
Gregory Kucherov\inst{2}
\and 
Tatiana Starikovskaya\inst{3}
}
\institute{
Higher School of Economics, \email{alexander.loptev@gmail.com}
\and
Laboratoire d'Informatique Gaspard Monge, Universit\'e Paris-Est \& CNRS, Marne-la-Vall\'ee, Paris, France, \email{gregory.kucherov@univ-mlv.fr}
\and
University of Bristol, Bristol, United Kingdom, \email{tat.starikovskaya@gmail.com} 
}
\date{\empty}

\maketitle

\begin{abstract}
Given a string $S$ of length $n$, its maximal unbordered factor is the longest factor which does not have a border. In this work we investigate the relationship between $n$ and the length of the maximal unbordered factor of $S$. We prove that for the alphabet of size $\sigma \ge 5$ the expected length of the maximal unbordered factor of a string of length~$n$ is at least $0.99 n$ (for sufficiently large values of $n$). As an application of this result, we propose a new algorithm for computing the maximal unbordered factor of a string. 
\end{abstract}

\section{Introduction}
If a proper prefix of a string is simultaneously its suffix, then it is called a border of the string. Given a string $S$ of length $n$, its maximal unbordered factor is the longest factor which does not have a border. The relationship between $n$ and the length of the maximal unbordered factor of $S$ has been a subject of interest in the literature for a long time, starting from the 1979 paper of Ehrenfeucht and Silberger~\cite{ESproblem}. 

Let $b(S)$ be the length of the maximal unbordered factor of $S$ and $\pi(S)$ be the minimal period of $S$. Ehrenfeucht and Silberger showed that if the minimal period of $S$ is smaller than $\frac{1}{2} n$, then $b(S) = \pi(S)$. Following this, they raised a natural question: How small $b(S)$ must be to guarantee $b(S) = \pi(S)$? Their conjecture was that $b(S)$ must be smaller than $\frac{1}{2} n$. However, this conjecture was proven false two years later by Assous and Pouzet~\cite{CounterExample}. As a counterexample they gave a string 

$$S = a^m b a^{m+1} b a^m b a^{m+2} b a^m b a^{m+1} b a^m$$
of length $n = 7m + 10$. The length of the maximal unbordered factor of this string is $b(S) = 3m + 6 \le \frac{3}{7} n + 2 < \frac{1}{2} n$ (with $ba^{m+1} ba^m ba^{m+2}$ and $a^{m+2} ba^m ba^{m+1} b$ being unbordered), and the minimal period $\pi (S) = 4m + 7 \neq b (S)$. 

The next attempt to answer the question was undertaken by Duval~\cite{Duval}: He improved the bound to $\frac{1}{4} n + \frac{3}{2}$. But the final answer to the question of Ehrefeucht and Silberger was given just recently by Holub and Nowotka~\cite{EhrenfeuchtSilberger-2}. They showed that $b(S) \le \frac{3}{7} n$ implies $b(S) = \pi(S)$, and, as follows from the example of Assous and Pouzet, this bound is tight.

Therefore, when either $b(S)$ or $\pi(S)$ is small, $b(S) = \pi(S)$. Exploiting this fact, one can even compute the maximal unbordered factor itself in linear time. The key idea is that in this case the maximal unbordered factor is an unbordered conjugate of the minimal period of $S$, and both the minimal period and its unbordered conjugate can be found in linear time~\cite{BorderArray,UnborderedConjugate}.

The interesting cases are those where $b(S)$ (and, consequently, $\pi(S)$) is big. Yet, it is generally believed that they are the most common ones. This is supported by experimental resuts shown in Fig.~\ref{fig:max_unbordered_length} that plots the average difference between the length~$n$ of a string and the length of its maximal unbordered factor. Guided by the experimental results, we state the following conjecture:

\begin{conjecture}
Expected length of the maximal unbordered factor of a string of length $n$ is $n - \Oh(1)$.
\end{conjecture}

\begin{figure}
\vspace*{-5pt}
\centering
\includegraphics[scale=0.5]{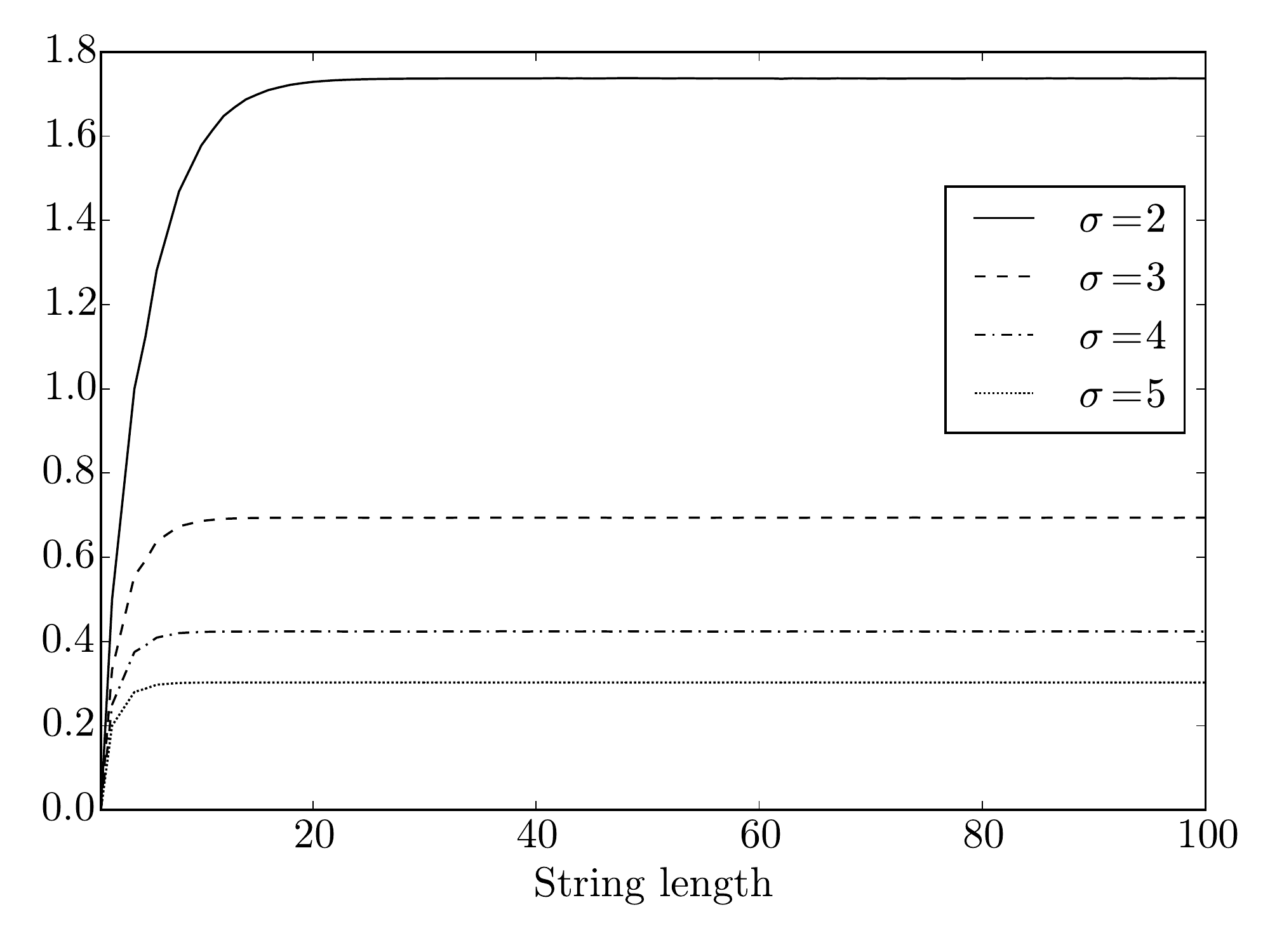}
\caption{Average difference between the length $n$ of a string and the length of its maximal unbordered factor for $1 \le n \le 100$ and alphabets of size $2 \le \sigma \le 5$.} 
\label{fig:max_unbordered_length}
\vspace*{-5pt}
\end{figure}

To the best of our knowledge, there have been no attempts to prove the conjecture or any lower bound at all in the literature. In Section~\ref{sec:lower_bound} we address this gap and make the very first step towards proving the conjecture. We show that the expected length of the maximal unbordered factor of a string of length $n$ over the alphabet $A$ of size $\sigma \ge 2$ is at least $n (1 - \xi(\sigma) \cdot \sigma^{-4}) + \Oh(1)$, where $\xi(\sigma)$ is a function that converges to $2$ quickly with the growth of $\sigma$. In particular, this theorem implies that for alphabets of size $\sigma \ge 5$ the expected length of the maximal unbordered factor of a string is at least $0.99 n$  (for sufficiently large values of $n$). To prove the theorem we developed a method of generating strings with large unbordered factors which we find to be interesting on its own (see Section~\ref{sec:generate}).

It follows that the algorithm for computing maximal unbordered factors we sketched earlier cannot be used in a majority of cases. Instead, one can consider the following algorithm. A border array of a string is an array containing the maximal length of a border of each prefix of this string. Note that a prefix of a string is unbordered exactly when the corresponding entry in the border array is zero. Therefore, to compute the maximal unbordered factor of a string $S$ it suffices to build border arrays of all suffixes of a string. It is well-known that a single border array can be constructed in linear time, which gives quadratic time bound for the algorithm. In  Section~\ref{sec:algo} we show how to modify this algorithm to make use of the fact that the expected length of the maximal unbordered factor is big. We give $\Oh(\frac{n^2}{\sigma^4})$ time bound for the modified algorithm, as well as confirm its efficiency experimentally.

\paragraph{Related work.} 
Apart from the aforementioned results, we consider our work to be related to three areas of research. 

As we have already mentioned, the maximal unbordered factor can be found by locating the rightmost zeros in the border arrays of suffixes of a string and better understanding of structure of border arrays would give more efficient algorithms for the problem. Structure of border arrays has been studied in~\cite{ValidatingKMP,ValidatingKMP-1,ValidatingKMP-2,ValidatingKMP-3,ValidatingKMP-4,NumberOfBorderArrays}. 

In contrast to the problem we consider in this work, one can be interested in the problem of preprocessing a string to answer online factor queries related to its borders. This problem has been considered by Kociumaka et al.~\cite{InternalPM,InternalPM-old}. They proposed a series of data structures which, in particular, can be used to determine if a factor is unbordered in logarithmic time.

Finally, repeating fragments in a string (borders of factors is one example of such fragments) were studied in connection with the \emph{Longest Common Extension} problem which asks, given a pair of positions $i,j$ in a string, to return the longest fragment that occurs both at $i$ and $j$. This problem has many solutions, yet recently Ilie at al.~\cite{LCE} showed that the simplest solution, i.e. simply scanning the string and comparing pairs of letters starting at positions $i$ and~$j$, is the fastest on average. The authors also proved that the longest common extension has expected length smaller than $\frac{1}{\sigma-1}$, where $\sigma$ is the size of the alphabet.

\section{Preliminaries}
\label{sec:prelim}
We start by introducing some standard notation and definitions. 

\paragraph{Power sums.} We will need the following identities.

\begin{fact}
\label{lm:power_sum}
$S(x) = \sum_{i=1}^{k}  i \; x^{i-1} = \frac{k \; x^{k+1} - (k+1) \; x^{k} + 1}{(x-1)^2}$ for all $x \neq 1$.
\end{fact}
\begin{proof}
$$S(x) = \bigl( \sum_{i=1}^{k} x^i \bigr) ' = \bigl( \frac{x^{k+1} - x}{x-1} \bigr) ' =  \frac{((k+1) x^{k} - 1)(x-1) - (x^{k+1} - x)}{(x-1)^2}$$
Simplifying, we obtain
$$S(x) = \sum_{i=1}^{k}  i \; x^{i-1} = \frac{k \; x^{k+1} - (k+1) \; x^{k} + 1}{(x-1)^2}$$
\qed
\end{proof}

\begin{corollary}
\label{cor:power_sum}
$S(x) = \sum_{i=1}^{k}  i \; x^{i-1} = \frac{k \; x^k}{x-1} + \Oh(x^{k-2})$ for $x \ge 1.5$.
\end{corollary}

\paragraph{Strings.} The alphabet $A$ is a finite set of size $\sigma$. We refer to the elements of~$A$ as \emph{letters}. A \emph{string} over $A$ is a finite ordered sequence of letters (possibly empty). Letters in a string are numbered starting from~1, that is, a string $S$ of \emph{length} $n$ consists of letters $S[1], S[2], \ldots, S[n]$. The length~$n$ of $S$ is denoted by~$|S|$. A set of all strings of length $n$ is denoted $A^n$.

For $1 \le i \le j \le n$, $S[i..j]$ is a \emph{factor} of $S$ with endpoints $i$ and $j$. The factor $S[1..j]$ is called a \emph{prefix} of $S$, and the factor $S[i..n]$ is called a \emph{suffix} of $S$. A prefix (or a suffix) different from $S$ and the empty string is called \emph{proper}.

If a proper prefix of a string is simultaneously its suffix, then it is called a \emph{border}. For example, borders of a string $ababa$ are $a$ and $aba$. The \emph{maximal border} of a string is its longest border. For $S$ we define its \emph{border array} $B$ (also known as the \emph{failure function}) to contain the lengths of the maximal borders of all its prefixes, i.e. $B[i]$ is the length of the maximal border of $S[1..i]$, $i = 1..n$. The last entry in the border array, $B[n]$, contains the length of the maximal border of $S$. It is well-known that the border array and therefore the maximal border of $S$ can be found in $\Oh(n)$ time and space~\cite{BorderArray}. 

A period of $S$ is an integer~$\pi$ such that for all $i$, $1 \le i \le n-\pi$, $S[i] = S[i +\pi]$. The minimal period of a string has length $n-B[n]$, and hence can be computed in linear time as well. 

\paragraph{Unbordered strings.} A string is called \emph{unbordered} if it has no border. Let $b(i, \sigma)$ be the number of unbordered strings in $A^i$. Nielsen~\cite{Bifixnote} showed that unbordered strings can be constructed in a recursive manner, starting from unbordered strings of length $2$ and inserting new letters in the ``middle''. The following theorem is a corollary of the proposed construction method:

\addtocounter{theorem}{-1}
\begin{theorem}[\cite{Bifixnote}]
\label{th:unbordered}
The sequence $\Big\{ \frac{b(i, \sigma)}{\sigma^i} \Big\}_{i = 1}^{\infty}$ is monotonically nonincreasing and it converges to a constant $\alpha$, which satisfies $\alpha \ge 1 - \sigma^{-1} - \sigma^{-2}$.
\end{theorem}

\begin{corollary}[\cite{Bifixnote}]
\label{cor:unbordered}
$b (i, \sigma) \ge \sigma^i - \sigma^{i-1} - \sigma^{i-2}$ for all $i$.
\end{corollary}

This corollary immediately implies that the expected length of the maximal unbordered factor of a string of length $n$ is at least $n (1 - \sigma^{-1}-\sigma^{-2})$. We improve this lower bound in the subsequent sections. We will make use of a lower bound on the number $b_j(i, \sigma)$ of unbordered strings such that its first letter differs from the subsequent $j$ letters. An example of such string for $j = 2$ is $a b c a c b b$.

\begin{lemma}
\label{lm:unbordered+}
$b_j (i, \sigma) \ge (\sigma-1)^{j+1}\sigma^{i-j-1} - \sigma^{i-2}$ for all $i \ge j+1$.
\end{lemma}
\begin{proof}
The number of such strings is equal to $b(i, \sigma)$ minus the number $b_j^- (i, \sigma)$ of unbordered strings of length $i$ that do not have the property. We estimate the latter from above by the number of such strings in the set of all strings with their first letter not equal to the last letter. Hence, $b_j^- (i, \sigma) \le (\sigma - 1) \sigma^{i-1} - (\sigma-1)^{j+1}\sigma^{i-j-1}$. Recall that $b(i, \sigma) \ge \sigma^i - \sigma^{i-1} - \sigma^{i-2}$ by Theorem~\ref{th:unbordered}. The claim follows.
\qed
\end{proof}

\begin{myremark}
The right-hand side of the inequality of Lemma~\ref{lm:unbordered+} is often negative for $\sigma = 2$. We will not use it for this case. 
\end{myremark}

The \emph{maximal unbordered factor} of a string (MUF) is naturally defined to be the longest factor of the string which is unbordered.

\section{Generating strings with large MUF}
\label{sec:generate}
In this section we explain how to generate strings of some fixed length $n$ with large maximal unbordered factors. To show the lower bounds we announced, we will need many of such strings. The idea is to generate them from unbordered strings.

Let $S$ be an unbordered string of length $i \ge \lceil\frac{n}{2}\rceil$. Consider a string $S P_1 \ldots P_k$ of length $n$, where $P_1, \ldots, P_k$ are prefixes of $S$. 
It is not difficult to see that the maximal unbordered factor of any string of this form has length at least~$i$. (Because $S$ is one of its unbordered factors.) The number of such strings that can be generated from $S$ is $2^{n-i-1}$, because each of them corresponds to a composition of $n-i$, i.e. representation of $n-i$ as a sum of a sequence of strictly positive integers. But, some of these strings can be equal. Consider, for example, an unbordered string $S = aaabab$. Then the two strings $aaababaaa$ ($S$ appended with its prefix $aaa$) and $aaababaaa$ ($S$ appended with its prefixes $a$ and $aa$) will be equal. However, we can show the following lemma.

\begin{lemma}
\label{lm:distinct}
Let $S_1 \neq S_2$ be two unbordered strings. Any two strings of the form above generated from $S_1$ and $S_2$ are distinct.
\end{lemma}
\begin{proof}
Suppose that the produced strings are equal. If $|S_1| = |S_2|$, we immediately obtain $S_1 = S_2$, a contradiction. Otherwise, w.l.o.g. assume $|S_1| < |S_2|$. Then $S_2$ is equal to a concatenation of $S_1$ and some of its prefixes. The last of these prefixes is simultaneously a suffix and a prefix of $S_2$, i.e. $S_2$ is not unbordered. A contradiction.
\qed
\end{proof}

Our idea is to produce as many strings of the form $S P_1 \ldots P_k$ as possible, taking extra care to ensure that all strings produced from a fixed string $S$ are distinct. From unbordered strings of length $i = n$ and $i = n-1$ we produce just one string of length $n$. (For $i = n$ it is the string itself and for $i = n-1$ it is the string appended with its first letter.) For unbordered strings of length $i \le n-2$ we propose a different method based on the lemma below.

\begin{lemma}
Each unbordered string $S$ of length $i$ such that its first letter differs from the subsequent $j$ letters, where $\lceil{n/2\rceil} \le i < n-j$, gives at least $2^j$ distinct strings of the form $S P_1 \ldots P_k$.
\end{lemma}
\begin{proof}
We choose the last prefix $P_k$ to be the prefix of $S$ of length at least $n-i-j$. We place no restrictions on the first $k-1$ prefixes. 

Let us start by showing that all generated strings are distinct. Suppose there are two equal strings $S P_1 \ldots P_\ell$ and $S P'_1 \ldots P'_{\ell'}$. Let $P_d, P'_d$ be the first pair of prefixes that have different lengths. W.l.o.g. assume that $|P_d| < |P'_d|$. Then $d \neq \ell$ and hence $|P_d| \le j = n - i - (n-i-j)$. It follows that $P'_d$ (which is a prefix of $S$) contains at least two occurrences of $S[1]$, one at the position $1$ and one at the position $|P_d|+1 \le j+1$. In other words, we have $S[1] = S[|P_d|+1]$ and $|P_d|+1 \le j+1$, which contradicts our choice of $S$.

If the length of the last prefix is fixed to some integer $m \ge n-i-j$, then each of the generated strings $S P_1 \ldots P_k$ is defined by the lengths of the first $k-1$ of the appended prefixes. In other words, there is one-to-one correspondence between the generated strings and compositions of $n-i-m$. (Here we use $i \ge \lceil{ n/2 \rceil}$ to ensure that every composition corresponds to a sequence of prefixes of~$S$.) The number of compositions of $n-i-m$ is $1$ when $m = n- i$ and $2^{n-i-m-1}$ otherwise. Summing up for all $m$ from $n-i-j$ to $n-i$ we obtain that the number of the generated strings is $2^j$.
\qed
\end{proof}

Let us estimate the total amount of strings produced by this method. We produce one string from each unbordered string of length $i$. Then, from each unbordered string of length $i$ such that its first letter differs from the second letter, we produce $1 = 2 - 1$ more string. If the first letter differs both from the second and the third letters, we produce $2 = 2^2 - 1 - 1$ more strings. And finally, if the first letter differs from the subsequent $j$ letters, we produce $2^{j-1} = 2^ j - \bigl(1 + 1 + 2 + \ldots + 2^{j-2}\bigr)$ strings. It follows that the number of strings we can produce from unbordered strings of length $i \le n-2$ is

\begin{equation*}
\label{eq:MUF=i}
b (i, \sigma) + \sum_{j=1}^{n-i-1} 2^{j-1} \cdot b_j (i, \sigma)
\end{equation*}
Recall that the maximal unbordered factor of each of the generated strings has length at least $i$ and that none of them can be equal to a string generated from an unbordered string of different length.

\section{Expected length of MUF}
\label{sec:lower_bound}
In this section we prove the main result of this paper. 

\begin{theorem}
\label{th:difference}
Expected length of the maximal unbordered factor of a string of length $n$ over an alphabet $A$ of size $\sigma \ge 2$ is at least 
\begin{equation}
n\cdot (1 - \xi(\sigma) \cdot \sigma^{-4}) + \Oh(1)
\end{equation}
where $\xi(2) = 8$ and $\xi(\sigma) = \frac{2\sigma^3 - 2\sigma^2}{(\sigma-2)(\sigma^2 - 2\sigma +2)}$ for $\sigma > 2$.
\end{theorem}

Before we give a proof of the theorem, let us say a few words about $\xi(\sigma)$. This function is monotonically decreasing for $\sigma \ge 2$ and quickly converges to $2$. We give the first four values for $\xi(\sigma)$ (rounded up to 3 s.f.) and $1 - \xi(\sigma) \cdot \sigma^{-4}$ (rounded down to 3 s.f.) in the table below.

\begin{table}
\centering
\begin{tabularx}{0.8\textwidth}{X|X|X|X|X}
 & $\sigma = 2$ &  $\sigma = 3$ & $\sigma = 4$ &  $\sigma = 5$ \\
 \hline
$\xi(\sigma)$ & 8.000 &  7.200 & 4.800 & 3.922\\
\hline
$1 - \xi(\sigma) \cdot \sigma^{-4}$ & 0.500 & 0.911 &  0.981 & 0.993
\end{tabularx}
\end{table}

\begin{corollary}
Expected length of the maximal unbordered factor of a string of length $n$ over the alphabet $A$ of size $\sigma \ge 5$ is at least $0.99 n$ (for sufficiently large values of $n$).
\end{corollary}

\subsubsection*{Proof of Theorem~\ref{th:difference}.} Let $\beta^n_i(\sigma)$ be the number of strings in $A^n$ such that the length of their maximal unbordered factor is $i$. Expected length of the maximal unbordered factor is then equal to

\begin{equation*}
\frac{1}{\sigma^n} \sum_{i = 1}^{n} i \cdot \beta^n_i(\sigma)
\end{equation*}
For the sake of simplicity, we temporarily omit $\frac{1}{\sigma^n}$, and only in the very end we will add it back. Recall that in the previous section we showed how to generate a set of distinct strings of length $n$ with maximal unbordered factors of length at least $i$ which contains

\begin{equation*}
b (i, \sigma) + \sum_{j=1}^{n-i-1} 2^{j-1} \cdot b_j (i, \sigma)
\end{equation*}
strings for all $\lceil \frac{n}{2}\rceil \le i \le n-2$ and $b(i, \sigma)$ strings for $i = \{n-1, n\}$. Then

\begin{equation}
\label{eq:start}
\sum_{i = 1}^{n} i \cdot \beta^n_i(\sigma) \ge \underbrace{\sum_{i=\lceil{n/2\rceil}}^{n} i \cdot b (i, \sigma)}_{(S_1)} + \underbrace{\sum_{i=\lceil{n/2\rceil}}^{n-2} \sum_{j=1}^{n-i-1} 2^{j-1} \cdot i \cdot b_j (i, \sigma)}_{(S_2)}
\end{equation}

We start by computing $(S_1)$.  Applying Corollary~\ref{cor:unbordered} and replacing $b(i, \sigma)$ with $\frac{b (n, \sigma)}{\sigma^{n-i}}$ in $(S_1)$, we obtain:
\begin{equation*}
(S_1) \ge \sum_{i = \lceil\frac{n}{2}\rceil}^{n} i \; \frac{b (n, \sigma)}{\sigma^{n-i}} = \frac{b (n, \sigma)}{\sigma^{n-1}} \bigl( \sum_{i = \lceil\frac{n}{2}\rceil}^{n} i \; \sigma^{i-1} \bigr)
\end{equation*}
Note that the lower limit in inner sum of $(S_1)$ can be replaced by one because the correcting term is small:
\begin{equation*}
\frac{b (n, \sigma)}{\sigma^{n-1}} \sum_{i=1}^{\lceil{n/2\rceil}-1}  i \sigma^{i-1} \le \frac{n^2 \cdot b(n, \sigma)}{4\sigma^{n/2}} = \Oh(\sigma^n)
\end{equation*}
We finally use Corollary~\ref{cor:power_sum} for $x = \sigma$ and $k=n$ to compute the right-hand side of the inequality: 
\begin{equation}
\label{eq:unbordered_as_period_fin}
(S_1) \ge \frac{n\sigma}{\sigma-1} \cdot b (n, \sigma) + \Oh(\sigma^n)
\end{equation}

We note that for $\sigma = 2$ the right-hand side is at least $2n \cdot (2^n - 2^{n-1}-2^{n-2})+\Oh(2^n) = n \cdot 2^{n-1} + \Oh(2^n)$ by Corollary~\ref{cor:unbordered} and $(S_2) \ge 0$. Hence, $\sum_{i = 1}^{n} i \cdot \beta^n_i(2) \ge n \cdot 2^{n-1} + \Oh(2^n)$. Dividing both sides by $2^n$, we obtain the theorem. 

Below we assume $\sigma > 2$ and for these values of $\sigma$ give a better lower bound on $(S_2)$. Recall that $b_j(i, \sigma) \ge (\sigma-1)^{j+1}\sigma^{i-j-1} - \sigma^{i-2}$ (see Lemma~\ref{lm:unbordered+}). It follows that
\begin{equation*}
(S_2) \ge \sum_{i=\lceil{n/2\rceil}}^{n-2} \sum_{j=1}^{n-i-1} 2^{j-1} \cdot i \cdot \bigl( (\sigma-1)^{j+1}\sigma^{i-j-1} - \sigma^{i-2} \bigr)
\end{equation*}
Let us change the order of summation:
\begin{equation*}
(S_2) \ge \sum_{j=1}^{\lfloor{n/2\rfloor}-1} 2^{j-1} \cdot \bigl(  (\sigma-1)^{j+1}\sigma^{-j} - \sigma^{-1} \bigr) \sum_{i=\lceil{n/2\rceil}}^{n-j-1}  i \cdot \sigma^{i-1}
\end{equation*}
We can replace the lower limit in the inner sum of $(S_2)$ by one as it will only change the sum by $\Oh(\sigma^n)$.
After replacing the lower limit, we apply Corollary~\ref{cor:power_sum} to compute the inner sum:
\begin{equation*}
(S_2) \ge \sum_{j=1}^{\lfloor{n/2\rfloor}-1} 2^{j-1} \cdot \bigl(  (\sigma-1)^{j+1} \sigma^{-j} - \sigma^{-1} \bigr) \cdot (n-j-1) \frac{\sigma^{n-j-1}}{\sigma-1} + \Oh(\sigma^n)
\end{equation*}
We divide the sum above into positive and negative parts:
\begin{equation*}
\underbrace{\sum_{j=1}^{\lfloor{n/2\rfloor}-1} (n-j-1) \; 2^{j-1} (\sigma-1)^{j} \sigma^{n-2j-1}}_{(P)} - \underbrace{\sum_{j=1}^{\lfloor{n/2\rfloor}-1} (n-j-1) 2^{j-1} \frac{\sigma^{n-j-2}}{\sigma-1}}_{(N)}
\end{equation*}
We start by computing $(N)$. We again apply the trick with the lower limit and Fact~\ref{lm:power_sum}, and replace $(n-j-1)$ with $k$.
\begin{equation*}
(N) = \frac{2^{n-3}}{\sigma-1} \sum_{k = \lceil{\frac{n}{2}\rceil}}^{n-2} k \bigr(\frac{\sigma}{2}\bigl)^{k-1} = \frac{(n-2)\sigma^{n-2}}{(\sigma-1)(\sigma-2)} + \Oh(\sigma^n)
\end{equation*}
Computing $(P)$ is a bit more involved. We divide it into two parts:
\begin{equation*}
(P) = \underbrace{\frac{(n-1) \sigma^{n-1}}{2} \cdot \sum_{j=1}^{\lfloor{n/2\rfloor}-1} \bigl( \frac{2 (\sigma-1)}{\sigma^2} \bigr)^j}_{R_1} - \underbrace{\sigma^{n-1} \sum_{j=1}^{\lfloor{n/2\rfloor}-1} j \; 2^{j-1} (\sigma-1)^j \sigma^{-2j}}_{R_2}
\end{equation*}
$(R_1)$ is a sum of a geometric progression and it is equal to
\begin{equation*}
\frac{(n-1) \sigma^{n-1}}{2} \cdot \frac{\bigl( \frac{2 (\sigma-1)}{\sigma^2} \bigr)^{\lfloor{n/2\rfloor}} - \frac{2 (\sigma-1)}{\sigma^2}}{\frac{2 (\sigma-1)}{\sigma^2} - 1} = \frac{(n-1) \sigma^{n-1}}{2} \cdot \frac{2(\sigma-1)}{\sigma^2 - 2\sigma + 2} + \Oh(\sigma^n)
\end{equation*}

\begin{lemma}
\label{lm:R_2_is_small}
$(R_2) = \Oh(\sigma^n)$.
\end{lemma}
\begin{proof}
We start our proof by rewriting $(R_2)$:
\begin{equation*}
(R_2) = \sigma^{n-3}(\sigma-1) \cdot \sum_{j=1}^{\lfloor{n/2\rfloor}-1} j \; \bigl( \frac{2 (\sigma-1)}{\sigma^2})^{j-1}
\end{equation*}
We apply Fact~\ref{lm:power_sum} for $x = \frac{2 (\sigma-1)}{\sigma^2}$ and $k = \lfloor{n/2\rfloor}-1$ to compute the inner sum.
\begin{equation*}
(R_2) = \sigma^{n-3}(\sigma-1) \cdot \frac{(\lfloor{n/2\rfloor}-1) \cdot (\frac{2 (\sigma-1)}{\sigma^2})^{\lfloor{n/2\rfloor}} - \lfloor{n/2\rfloor} \cdot (\frac{2 (\sigma-1)}{\sigma^2})^{\lfloor{n/2\rfloor}-1} + 1}{(\frac{2 (\sigma-1)}{\sigma^2} - 1)^2}
\end{equation*}
The claim follows.
\qed
\end{proof}

We now summarize our findings. From equations for $(P)$, $(N)$, $(R_1)$, and $(R_2)$ we obtain (after simplification):
\begin{equation}
\label{eq:unbordered+prefixes}
(S_2) \ge (P) - (N) = n \cdot \bigl( \frac{\sigma^n - \sigma^{n-1}}{\sigma^2 - 2\sigma + 2} - \frac{\sigma^{n-2}}{(\sigma-1)(\sigma-2)}\bigr) + \Oh(\sigma^n) 
\end{equation}

We now return back to Equation~\eqref{eq:start} and use our lower bounds for $(S_1)$ and $(S_2)$ together with Corollary~\ref{cor:unbordered} for $b(n, \sigma)$:
\begin{equation*}
\sum_{i = 1}^{n} i \cdot \beta^n_i(\sigma) \ge n \cdot \bigl( \frac{\sigma^{n+1} - \sigma^n - \sigma^{n-1}}{\sigma - 1} + \frac{\sigma^n - \sigma^{n-1}}{\sigma^2 - 2\sigma + 2} - \frac{\sigma^{n-2}}{(\sigma-1)(\sigma-2)} \bigr) + \Oh(\sigma^n) 
\end{equation*}

We now simplify the expression above and return back $\frac{1}{\sigma^n}$ as we promised in the very beginning of the proof to obtain:
\begin{equation}
\frac{1}{\sigma^n} \sum_{i=1}^n i \cdot \beta^n_i (\sigma) \ge n\cdot (1 - \xi(\sigma) \cdot \sigma^{-4}) + \Oh(1)
\end{equation}
where $\xi(\sigma) = \frac{2\sigma^3 - 2\sigma^2}{(\sigma-2)(\sigma^2 - 2\sigma +2)}$. This completes the proof of Theorem~\ref{th:difference}.
\qed

\begin{myremark}
Theorem~\ref{th:difference} actually provides a lower bound on the expected length of the maximal unbordered prefix (rather than that of the maximal unbordered  factor), which suggests that this bound could be improved.
\end{myremark}

\section{Computing MUF} 
\label{sec:algo}
Based on our findings we propose an algorithm for computing the maximal unbordered factor of a string $S$ of length $n$ and give an upper bound on its expected running time. A basic algorithm would be to compute the border arrays (see Section~\ref{sec:prelim} for the definition) of all suffixes of $S$. The border arrays contain the lengths of the maximal borders of all prefixes of all suffixes of $S$, i.e., of all factors of $S$. It remains to scan the border arrays and to select the longest factor such that the length of its maximal border is zero. Since a border array can be computed in linear time, the running time of this algorithm is $\Oh(n^2)$.

The algorithm we propose is a minor modification of the basic algorithm. We build border arrays for suffixes of $S$ starting from the longest one. After building an array $B_i$ for $S[i..n]$ we scan it and locate the longest factor $S[i..j]$ such that the length of its maximal border stored in $B_i [j]$ is zero. We then compare $S[i..j]$ and the current maximal unbordered factor (initialized with an empty string). If $S[i..j]$ is longer, we update the maximal unbordered factor and proceed. At the moment we reach a suffix shorter than the current maximal unbordered factor, we stop. 

\begin{theorem}
\label{th:avg_time_naive}
The maximal unbordered factor of a string of length $n$ over an alphabet $A$ of size $\sigma$ can be found in $\Oh(\frac{n^2}{\sigma^4})$ expected time.
\end{theorem}
\begin{proof}
Let $b(S)$ be the length of the maximal unbordered factor of $S$. Then the running time of the algorithm is $\Oh((n-b(S))\cdot n)$, because $b(S)$ will be a prefix of one of the first $n-b(S)+1$ suffixes of $S$ (starting from the longest one). Averaging this bound over all strings of length $n$, we obtain that the expected running time is

$$\Oh(\frac{1}{\sigma^n} \sum_{S\in A^n} (n-b(S)) \cdot n ) = \Oh(n \cdot ( \frac{1}{\sigma^n} \sum_{S\in A^n} (n-b(S))))$$
and $\frac{1}{\sigma^n} \sum_{S\in A^n} (n-b(S)) = \Oh(\frac{n}{\sigma^4})$ as it follows from Theorem~\ref{th:difference} and properties of $\xi(\sigma)$. 
\qed
\end{proof}

We performed a series of experiments to confirm that the expected running time of the proposed algorithm is much smaller than that of the basic algorithm. We compared the time required by the algorithms for strings of length $1 \le n \le 100$ over alphabets of size $\sigma=\{2, 3, 4, 5, 10\}$. The time required by the algorithms was computed as the average time on a set of size $N = 10^6$ of randomly generated strings of given length. The experiments were performed on a PC equipped with one 2.6 GHz Intel Core i5 processor. As it can be seen in Fig.~\ref{fig:running_time}, the minor modification we proposed decreases the expected running time dramatically. Obtained results were similar for all considered alphabet sizes. All source files, results, and plots can be found in a repository \texttt{\url{http://github.com/avlonger/unbordered}}. 

\begin{figure}
\centering
\includegraphics[scale=0.4]{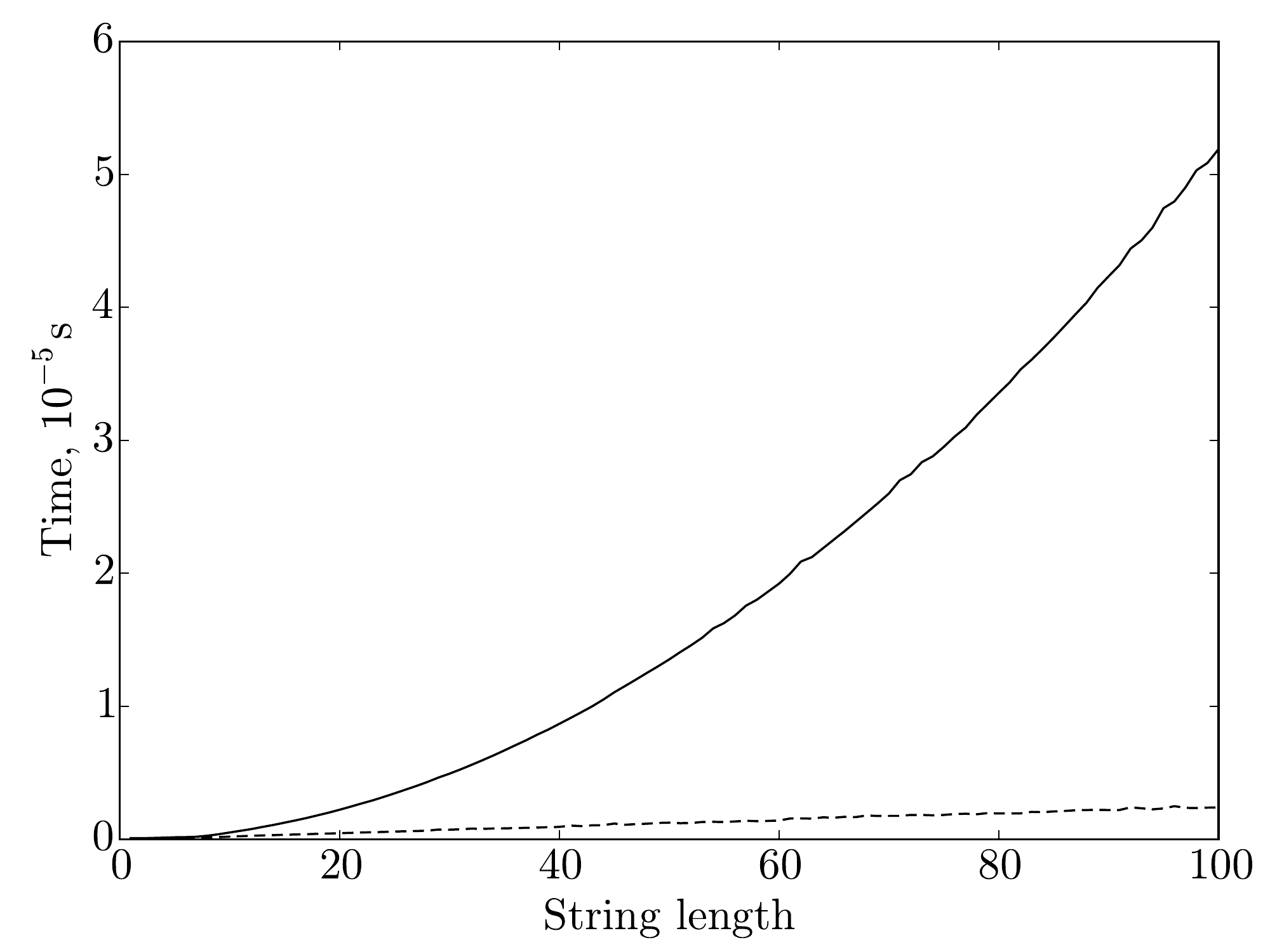}
\caption{Average running times of the proposed algorithm (dashed line) and the basic algorithm (solid line) for strings over the alphabet of size $\sigma = 2$.}
\label{fig:running_time}
\end{figure}

We note that the data structures~\cite{InternalPM,InternalPM-old} can be used to compute the maximal unbordered factor in a straightforward way by querying all factors in order of decreasing length. This idea seems to be very promising since these data structures need to be built just once, for the string $S$ itself. However, the data structures are rather complex and both the theoretical bound for the expected running time, which is $\Oh(\frac{n^2}{\sigma^4} \log n)$, and our experiments show that this solution is slower than the one described above. 

\section{Conclusion}
We consider the contributions of this work to be three-fold. We started with an explicit method of generating strings with large unbordered factors. We then used it to show that the expected length of the maximal unbordered factor and the minimal period of a string of length~$n$ is $\Omega(n)$, leaving the question raised in Conjecture 1 open. As an immediate application of our result, we gave a new algorithm for computing maximal unbordered factors and proved its efficiency both theoretically and experimentally. 

\subsection*{Acknowledgements} The authors would like to thank the anonymous reviewers whose suggestions greatly improved the quality of this work. 

\bibliographystyle{plain}
\bibliography{main}
\end{document}